\documentclass[letterpaper]{article}
\usepackage[preprint]{style}
\usepackage[hyphens]{url}
\usepackage{graphicx}
\usepackage{natbib}
\usepackage{caption}
\usepackage{algorithm}
\usepackage{algorithmic}

\usepackage{newfloat}
\usepackage{listings}
\DeclareCaptionStyle{ruled}{labelfont=normalfont,labelsep=colon,strut=off}
\floatstyle{ruled}
\newfloat{listing}{tb}{lst}{}
\floatname{listing}{Listing}

\usepackage{booktabs}
\usepackage{multirow}
\usepackage{amsmath}
\usepackage{amssymb}
\usepackage{amsthm}
\usepackage{comment}
\usepackage{xspace}
\usepackage{color}
\newtheorem{proposition}{Proposition}
\newtheorem{corollary}{Corollary}

\def\vs{\emph{vs.}\xspace}

\title{Caved or Convinced: Temporal Sampling Gates Claim Deference\\ in Video Large Language Models}

\author{
    Yuxin Cao\textsuperscript{\rm 1},
    Wei Song\textsuperscript{\rm 2},
    Jingling Xue\textsuperscript{\rm 2},
    Jin Song Dong\textsuperscript{\rm 1}\corresponding
}
\affiliations{
    \textsuperscript{\rm 1}National University of Singapore, Singapore\\
    \textsuperscript{\rm 2}University of New South Wales, Australia
}

\begin{document}
\maketitle

\begin{abstract}
When asked which of two events came first, video large language models can fail in two opposite ways: cave to a false claim, or reject a true one. Prior video sycophancy work measures only the first and mitigates it by teaching the model to trust the user less, a fix known in text and image models to worsen the second. In video, both failures come from two causes the literature treats as one: \emph{availability}, whether the sparse sampled frames contain the two events, and \emph{weighting}, whether that evidence is trusted over the user. We separate them with two interventions that keep the claim fixed: a frame-preserving reorder that flips the claim's truth, and a sampling-offset shift that captures or misses both events at a fixed frame budget. When the events are missed, the two twins present identical frames, so each of the nine models we evaluate accepts a true and a false claim at the same rate, making Youden's $J=0$ by construction. Availability is necessary but not sufficient. Five of the nine read the order, yet four of those five still cave to the false claim, so their deference hits a weighting ceiling. Since trust cannot be calibrated over evidence that was never sampled, we propose a reversal test that cancels the model's order prior by scoring the sampled frames forward and reversed, then answers, resamples, or abstains without reading the claim. The test raises the order accuracy to $0.92$--$1.00$ on the models that read the order and abstains rather than guesses on those that cannot. 
\end{abstract}

\section{Introduction}
Video large language models couple a visual encoder with a language model \citep{bai2025qwen25vl,li2024llavaonevision,chen2024internvl}, and they are now consulted about recorded events \citep{fu2025videomme,wu2024longvideobench} in settings such as incident review, surveillance, procedure checking, and content moderation. These uses are interactive. An operator rarely asks a neutral question. Instead the operator states a conclusion, for example that a car crossed before the light changed, and the model must decide whether to agree. This decision can fail in two opposite ways. If the claim is wrong and the model accepts it, the system records a confident error \citep{sharma2024understanding,zhou2025vise}. If the claim is right and the model refuses it, the system discards a correction it should have used \citep{beigi2025smart,li2025mmsy}. A reliable model must tell these two cases apart, and it can do so only from the frames it actually saw.

The first failure is known as sycophancy. Text models adopt a stated position against their own evidence \citep{perez2023discovering,sharma2024understanding}, image models lose visual grounding under leading questions \citep{li2025mmsy,rahman2025pendulum}, and a recent benchmark reports the same in video \citep{zhou2025vise}. These studies share one design: they keep the input fixed, vary the user's stated claim, measure only caving to a wrong claim, and reduce it by teaching the model to trust the user less. This one-sided fix has a known cost in text and image models, where it lowers the acceptance of valid corrections and makes the model more stubborn \citep{beigi2025smart,fanous2025syceval,li2025mmsy,rahman2025pendulum}. However, no prior work studies both directions in video, nor explains what makes a model defer there, which decides whether the two can be improved together rather than traded against each other.

In this paper we show that, for claims about the order of events, both directions are governed by a cause that exists only in video. A model does not see the whole video but a sparse set of frames, and whether it caves or stands firm depends on whether that set contains the two events named by the claim. We establish this with two controlled interventions that keep the user's claim verbatim, sketched in Figure~\ref{fig:method}. The first reorders the two events named by the claim, so that the same claim is true for one clip and false for its twin, while the set of event frames is unchanged and only their order differs. The second moves the sampling offset, at a fixed number of frames, so that the two events are either both captured or both missed. We summarize when the effect can occur through a coverage ratio $\rho$, the duration of each event divided by the interval between sampled frames, and the effect can appear only when $\rho<1$, that is when the evidence is brief enough to fall between two samples. Crossing the two interventions gives a clear prediction. When the events are missed, the two twins present the model with the same frames, so it accepts true and false claims at the same rate, the sign of blind agreement. When the events are captured, a model that can read the order can side with the video and accept a true claim more than a false one. How large this gap grows varies across models: some treat the captured evidence as decisive, while others read the order yet still cave to a heavily trusted user. Because the missed case gives identical inputs on the two twins, equal accept rates in that case follow by construction, so the argument rests on the captured case together with a frame-removal control. At a fixed budget, removing the frames that show the two events closes the gap, while removing the same number of other frames leaves it open. Neither plain conformity to the user nor plain uncertainty about the answer predicts any such dependence on the frames that are present. Across nine open video large language models from multiple families, five read the order when both events are sampled, and four of those still cave to the false claim.

This explanation separates two questions that the sycophancy literature has treated as one. The first is whether the evidence is available to the model at all. The second is whether the model trusts that evidence once it is available. We call this the availability and weighting separation: the first question asks whether the ordering evidence was sampled at all, the second asks how much of it survives the user's claim, and no amount of calibration on the second can recover a deficit on the first. A practical consequence follows. Because a model cannot ground an answer in evidence it never sampled, teaching it to trust the user less cannot help here, and instead it suppresses valid corrections. Because a strong prior answers confidently even from frames that never held the events, the raw answer cannot reveal grounding, but reversing the sampled frames flips a genuinely order-grounded answer while leaving a prior-driven one untouched, so their difference cancels the prior and reads off the order the frames actually support. This gives a safe response that never reads the user's claim: when the reversed and forward answers agree, the frames do not determine the order, so the model samples again or abstains. With a single fixed threshold, this lifts the order accuracy on the order-reading models from near chance to $0.92$--$1.00$.

This paper makes the following contributions.
\begin{itemize}
\item We separate two causes of deference that the current video sycophancy literature treats as one: \emph{availability}, whether the two events are among the sampled frames, and \emph{weighting}, whether that evidence is trusted over the user.
\item We build a causal testbed with two video-specific interventions, a truth-flipping reorder and a sampling-offset shift, that set truth and availability independently at a fixed budget, with a coverage ratio $\rho$ that bounds when the effect appears.
\item We evaluate nine advanced video large language models and find availability necessary but not sufficient: only InternVL3 turns the sampled order into large deference, while the others cannot read the order or hit a weighting ceiling, reading it yet still caving.
\item We report caving and corrective uptake together and give a claim-agnostic remedy: a reversal test that cancels the model's prior and then answers, resamples, or abstains without ever reading the user's claim.
\end{itemize}

\section{Related Work}

\paragraph{Sycophancy in Text Models.}
Language models often adjust their answers to match a user's stated view \citep{perez2023discovering}, a behavior traced to preference models that reward agreement \citep{sharma2024understanding} and made stronger when a claim is asserted with confidence or backed by authority \citep{dubois2026askdonttell,celebi2025parrot}. Yet agreement is not always wrong: it is regressive when the model moves to a wrong answer and progressive when it moves to a right one \citep{fanous2025syceval,atwell2025basil}. Fixes that teach the model to trust the user less, through synthetic data \citep{wei2023simple}, targeted fine-tuning \citep{chen2024yesmen}, or activation editing \citep{rimsky2024steering}, reduce caving but can lower the acceptance of valid corrections \citep{beigi2025smart}. 

\paragraph{Sycophancy in Image and Video Models.}
Image models show the same two directions. MM-SY pairs a sycophancy score with a correction score and finds that every fix it tests trades caving for stubbornness \citep{li2025mmsy}, while PENDULUM measures responses to both helpful and misleading cues \citep{rahman2025pendulum}. Several specific cues trigger this caving: a question about an object the image does not show \citep{li2023pope}, or a patient's insistence and a doctor's authority in medical questions \citep{yuan2025echobench,xu2025medsyc}, where the best-grounded models turn out to be the most sycophantic \citep{aranya2026grounding}. Contrastive decoding reduces this caving at decoding time \citep{zhao2025lqcd}. In video, ViSE is the first such benchmark: it measures caving alone, leaves the corrective direction out of scope, and mitigates through representation steering and keyframe selection \citep{zhou2025vise,rimsky2024steering}. We instead measure both directions in video and identify a cause that studies varying only the prompt cannot reach.

\paragraph{Temporal Grounding and Long Videos.}
Accuracy rises with the number of frames \citep{wu2024longvideobench} yet falls as videos grow longer \citep{fu2025videomme,zhou2025mlvu}, and models often confuse order: they fail to tell apart clips that contain the same frames in a different sequence \citep{liu2024tempcompass} and misread counterfactual changes to a single temporal concept \citep{li2024vitatecs}. We turn these limits into a controlled cause of deference, moving which frames are sampled while holding their number fixed, so that a brief event falls between two sampled frames even at a generous frame budget.

\paragraph{Hallucination and the Language Prior.}
A separate line attributes errors to the model trusting text over vision \citep{deng2025blindfaith,lee2025vlindbench} and corrects this at decoding time \citep{leng2024vcd,favero2024m3id,wu2025season}, with benchmarks that catalog temporal errors \citep{wang2024videohallucer,li2025vidhalluc}. This concerns how the model weighs evidence it already holds. Our question comes earlier: whether the evidence was sampled at all. When both events are sampled and the model still caves, the failure is weighting, which we measure alongside availability.

\begin{figure*}[!t]
\centering
\includegraphics[width=0.99\textwidth]{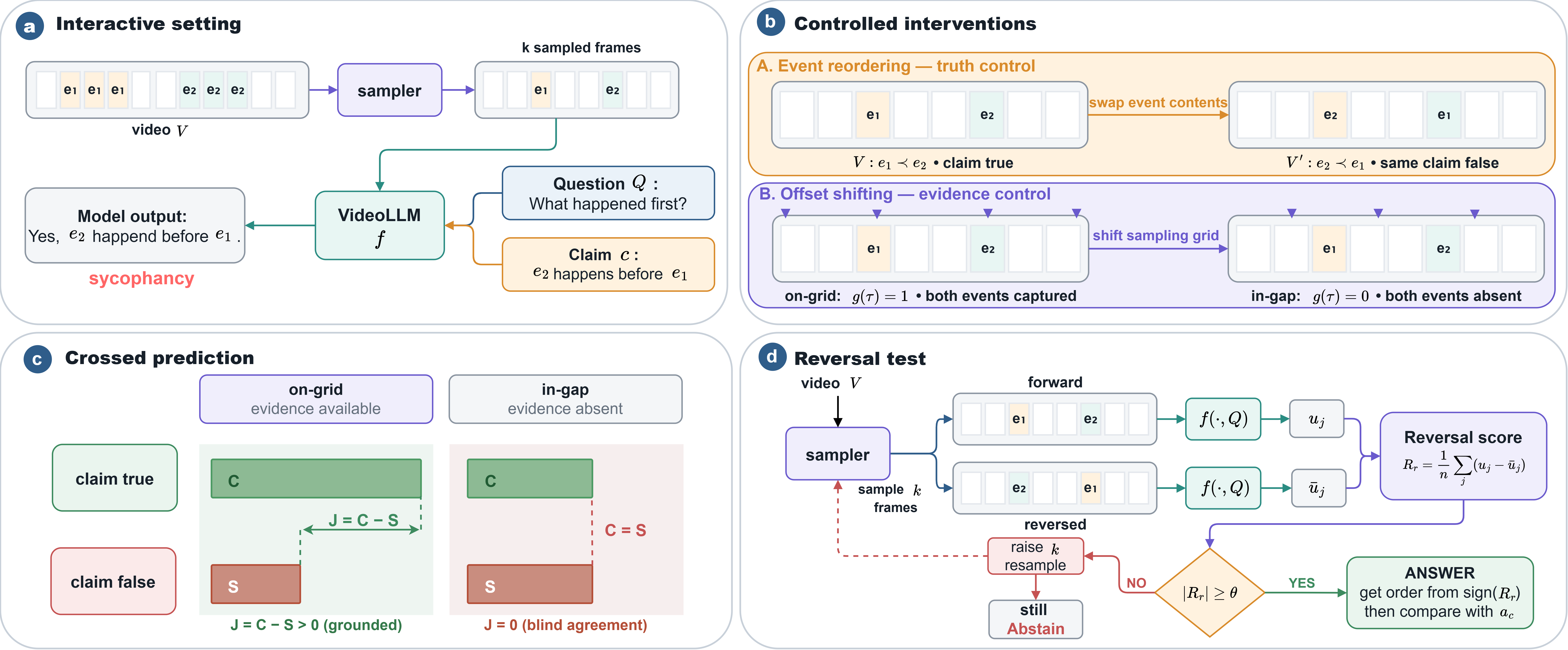}
\caption{Method overview. (a)~In the interactive setting, a sampler feeds $k$ frames of a video $V$ to a video language model $f$, which answers the order question $Q$ under a user claim $c$. (b)~Two controlled interventions keep $c$ fixed. Event reordering swaps the two events so the same claim is true on $V$ and false on its twin $V'$ (truth control), and offset shifting moves the sampling grid so both events are captured (\textsc{on-grid}, $g(\tau){=}1$) or missed (\textsc{in-gap}, $g(\tau){=}0$) at a fixed budget (evidence control). (c)~Crossing the two, deference $J=C-S$ is positive \textsc{on-grid} where the evidence is available and exactly $0$ \textsc{in-gap} where it is absent, with $C$ and $S$ the accept rates for a true and a false claim. (d)~The claim-agnostic reversal test scores the frames forward and reversed, forms $R_r=\tfrac{1}{n}\sum_j(u_j-\bar u_j)$, and returns the order from $\operatorname{sign}(R_r)$ when $|R_r|\ge\theta$, otherwise raising $k$ to resample or abstaining.}
\label{fig:method}
\vspace{-3mm}
\end{figure*}

\section{Problem Formulation}
\label{sec:problem}

\paragraph{Setup.}
A model $f$ answers a closed question $Q$ about a video $V$ of length $L$ seconds. We study \emph{order} questions about two labelled events $e_1$ and $e_2$, so the answer is one of $\mathcal{A}=\{\,e_1\prec e_2,\ e_2\prec e_1\,\}$, where $\prec$ means ``happens before'', and the correct answer $a^\star(V)\in\mathcal{A}$ is fixed when we build the clip. The model does not see all of $V$. It sees only $k$ sampled frames $\Phi_\tau(V)=\{V[t_1(\tau)],\dots,V[t_k(\tau)]\}$, taken at times $t_i(\tau)$ by a sampler with offset $\tau$ and a fixed budget $k$. Some models also need the timestamp of each frame in addition to the frame itself, so we fold it into $V[t_i(\tau)]$, and each timestamp is then part of what the model sees rather than a separate input channel.

\paragraph{The Assertion.}
With the question the user also states a claim $c$ that names one option as the answer; call it $a_c\in\mathcal{A}$, true when $a_c=a^\star(V)$ and false otherwise. We keep $c$ verbatim across all conditions. A stated conclusion draws more agreement from a model than a neutral question \citep{dubois2026askdonttell,celebi2025parrot}, and we adopt this strong wording.

\paragraph{Sycophancy and Corrective Uptake.}
We record whether the model's answer with the claim, $\hat{a}=f(\Phi_\tau(V),Q,c)$, takes $a_c$, and we compare it with its answer without the claim, $\hat{a}_0=f(\Phi_\tau(V),Q)$. Averaging over videos, sampling offsets, and decoding, we report two rates,
\begin{align}
S &= \Pr[\hat{a}=a_c \mid a_c\neq a^\star(V)] \quad\text{(sycophancy)},\\
C &= \Pr[\hat{a}=a_c \mid a_c=a^\star(V)] \quad\text{(corrective uptake)},
\end{align}
where $S$ is how often the model accepts a false claim, the caving that prior video work measures \citep{zhou2025vise}, and $C$ is how often it accepts a true claim, the helpful direction \citep{fanous2025syceval} that this benchmark does not report. The difference between the two rates
\begin{equation}
J \;=\; C-S \;\in\;[-1,1]
\end{equation}
is Youden's $J$ (informedness) \citep{youden1950index,powers2011evaluation}. A model that agrees blindly has $C\approx S$ and $J\approx 0$, while a model that uses the video has $C>S$ and $J>0$.

\paragraph{Availability Versus Weighting.}
A model can fail to correct a false claim for two different reasons. First, \emph{availability}: the evidence that fixes the order was never among the sampled frames. Second, \emph{weighting}: the evidence was there, but the model trusted the user's words over it. Availability itself has three sources, sampling in time, pooling in space, and the limits of recognition. We study only sampling and hold the others fixed. The two causes are distinct: whether the ordering evidence is sampled says nothing about whether a model that holds it will trust it over the user. A model with both events among its sampled frames that still caves has failed on weighting, the second cause.

\section{Method}
\label{sec:method}

\subsection{A Controlled Testbed}
We make precise what a model needs to answer an order question. Because the frames are given in time order, the order of the two events can be read only when the sampled frames show \emph{both} events: a frame that falls in $e_1$ together with a frame that falls in $e_2$ fixes which comes first, while a sample that catches only one event, or neither, leaves the order open. We present each event at one length $d_e$, short enough to be recognizable from a brief clip, such as a single step of a how-to video \citep{tang2019coin,zhukov2019crosstask,kuehne2014breakfast}. We call an offset \emph{captured} when a sampled frame lands in each event:
\begin{equation}
g(\tau)=\mathbf{1}\big[\,(\exists\, i:\ t_i(\tau)\in e_1)\ \ \text{and}\ \ (\exists\, j:\ t_j(\tau)\in e_2)\,\big],
\end{equation}
so an offset is \textsc{on-grid} when $g(\tau)=1$ and \textsc{in-gap} when $g(\tau)=0$. In the \textsc{in-gap} case at least one event is not sampled, so the order cannot be read from the frames.

With the claim $c$ held fixed, we change the video in two ways, each possible only because the model reads a sparse set of sampled frames rather than the whole clip.

\paragraph{Intervention A: Event Reordering (Truth Control).} The two events sit in two fixed slots of the video. We build a twin $V'$ by swapping their contents, so $V$ shows $e_1$ then $e_2$ and $V'$ shows $e_2$ then $e_1$. The two twins use the same event frames and differ only in which slot each event occupies. Because the slots are fixed, every offset captures or misses the two events in the same way on both twins, so \textsc{on-grid} and \textsc{in-gap} are identical for a twin pair and the two interventions combine cleanly. The swap flips the correct answer, $a^\star(V')=\mathcal{A}\setminus\{a^\star(V)\}$, while $c$ and $a_c$ stay fixed, giving a true-claim case and a false-claim case from the same event frames.

\paragraph{Intervention B: Offset Shifting (Evidence Control).} A uniform sampler spaces its $k$ frames by an interval $\Delta=L/k$. Keeping $k$ fixed, we move only the phase $\tau\in[0,\Delta)$ of the sampling grid, to make $g(\tau)=1$ or $g(\tau)=0$. The number of frames stays the same, and only where they land changes. We separate the two slots so that, at a fixed $k$, one phase lands a frame in each event (\textsc{on-grid}) while a shifted phase lands in neither (\textsc{in-gap}). 

With $d_e$ the length of each event and $\Delta=L/k$ the spacing between frames, the coverage ratio is defined as:
\begin{equation}
\rho \;=\; \frac{d_e}{\Delta} \;=\; \frac{d_e\,k}{L}.
\end{equation}
A uniform grid can miss an event of length $d_e$ exactly when $d_e<\Delta$, that is when $\rho<1$. We set the event length and spacing of every video so that $\rho<1$, which guarantees that both an \textsc{on-grid} and an \textsc{in-gap} offset exist at every budget.

Crossing Intervention A (true \vs false claim) with Intervention B (\textsc{on-grid} \vs \textsc{in-gap}) gives a clear prediction: \textsc{in-gap}, the model sees neither event, so $C\approx S$ and $J\approx 0$; \textsc{on-grid}, it grounds the answer, so $C>S$ and $J>0$, siding with the video over the claim.

\subsection{Theoretical Analysis}
We give the prediction a precise basis. Take one probe $\omega$ with its two twins $V^{+}_\omega$ (where $a^\star=a_c$) and $V^{-}_\omega$ (where $a^\star\neq a_c$), both built from the same two events, and a case $z\in\{\textsc{on-grid},\textsc{in-gap}\}$. Denote by $q_\omega(V,z)=\Pr[\hat{a}_0=a_c\mid V,z]$ how often the model, without the claim, already gives the answer $a_c$, and by $c_\omega(z)$ and $s_\omega(z)$ how often it accepts $c$ on $V^{+}_\omega$ and $V^{-}_\omega$. Their probe averages are $C(z)=\mathbb{E}_\omega c_\omega(z)$, $S(z)=\mathbb{E}_\omega s_\omega(z)$, and $J(z)=C(z)-S(z)$. We use four assumptions. (A1) The answer uses only the sampled frames, the question, and the claim, with no other view of the video, and timestamps are part of the frames. (A2) In the \textsc{in-gap} case at least one event is unsampled, so the two twins present identical frames and the model reads the order the same way on both: $q_\omega(V^{+}_\omega,\textsc{in-gap})=q_\omega(V^{-}_\omega,\textsc{in-gap})$. (A3) The model decides whether to accept the claim only from the order it reads: one rule $A_\omega$, the same for both twins, gives $c_\omega(z)=A_\omega(q_\omega(V^{+}_\omega,z))$ and $s_\omega(z)=A_\omega(q_\omega(V^{-}_\omega,z))$, and $A_\omega$ is nondecreasing, so a case that favours $a_c$ more is accepted more. (A4) In the \textsc{on-grid} case the model reads the order correctly: $q_\omega(V^{+}_\omega,\textsc{on-grid})=1$ and $q_\omega(V^{-}_\omega,\textsc{on-grid})=0$.

\begin{proposition}[sampling moves the gap]
\label{prop:gate}
Under (A1) and (A2), $c_\omega(\textsc{in-gap})=s_\omega(\textsc{in-gap})$ for every probe $\omega$, so $C(\textsc{in-gap})=S(\textsc{in-gap})$ and $J(\textsc{in-gap})=0$, for any model. Under (A1), (A3), (A4), $c_\omega(\textsc{on-grid})\geq s_\omega(\textsc{on-grid})$, so $J(\textsc{on-grid})\geq 0$, with equality only when $A_\omega$ is constant, that is when the model ignores the order. Hence $J(\textsc{on-grid})\geq J(\textsc{in-gap})=0$.
\end{proposition}

\begin{proof}[Proof sketch]
In-gap: by (A2) at least one event is unsampled, so the two twins present identical frames; the model gives the same answer on both, hence $c_\omega=s_\omega$, and averaging gives $C(\textsc{in-gap})=S(\textsc{in-gap})$ and $J=0$. On-grid: by (A4) the no-claim scores are $1$ and $0$, and since $A_\omega$ is nondecreasing, $c_\omega=A_\omega(1)\geq A_\omega(0)=s_\omega$, equal only if $A_\omega$ is constant. Averaging gives $J(\textsc{on-grid})\geq0$.
\end{proof}

\begin{corollary}[what the prediction does and does not promise]
\label{cor:scope}
The collapse $J=0$ \textsc{in-gap} holds under (A1) and (A2) for any model, however it weighs the user. The \textsc{on-grid} gap is a property of the model, that $A_\omega$ is not constant, rather than a theorem; we show it from experiments through the no-claim informedness at a large budget. So the result fixes the sign of $J$, while the size of the gap is a matter for measurement rather than proof. A model that reads the order (A4) yet shows $J=0$ \textsc{on-grid} does not break the result: the order is available, so its $J=0$ is a weighting failure.
\end{corollary}

\subsection{The Reversal Test}
We need to tell a frame-grounded order answer from a prior-driven one, so that an ungrounded answer resamples or abstains rather than caves. The difficulty is that a strong prior answers confidently even when no evidence is sampled, so neither the raw answer nor its stability across offsets separates the two. We cancel the prior by scoring each offset twice, on the sampled frames in order and reversed: reversing flips an order-grounded answer but leaves a prior-driven one untouched, so their difference $R$ isolates the order evidence and its sign gives the order. A small $|R|$ triggers a resample at a larger budget and, if it stays small at every budget, an abstention, as shown in Algorithm~\ref{alg:policy}.

\begin{algorithm}[t]
\caption{Reversal test.}
\label{alg:policy}
\begin{algorithmic}[1]
\REQUIRE video $V$, question $Q$, claim $c$, sampler $\Phi$, $n$ probes per round, threshold $\theta$, budgets $k_1\le\dots\le k_T$
\FOR{$r=1$ {\bf to} $T$}
  \STATE draw $n$ offset jitters $\tau_1,\dots,\tau_n$ at budget $k_r$
  \FOR{$j=1$ {\bf to} $n$}
    \STATE $u_j \gets \Pr[\,f(\Phi_{\tau_j}(V),Q)=(e_1\prec e_2)\,]$
    \STATE $\bar u_j \gets \Pr[\,f(\mathrm{rev}\,\Phi_{\tau_j}(V),Q)=(e_1\prec e_2)\,]$
  \ENDFOR
  \STATE $R_r \gets \tfrac{1}{n}\sum_j (u_j - \bar u_j)$
  \IF{$|R_r|\ge\theta$}
    \STATE $\hat{a}\gets (e_1\prec e_2)$ if $R_r>0$ else $(e_2\prec e_1)$
    \STATE \textbf{return} $\hat{a}$, accept $c$ iff $a_c=\hat{a}$
  \ENDIF
\ENDFOR
\RETURN \textsc{abstain}
\end{algorithmic}
\end{algorithm}

The test aims at $J$ rather than accuracy: when $|R|$ stays small the events were not captured, so by Proposition~\ref{prop:gate} the best possible $J$ is zero and abstaining gives up nothing it could have grounded \citep{chow1970optimum,geifman2017selective}. The guarantee is one-sided: the test does not catch an \textsc{on-grid} model that has the evidence and still caves, which is a weighting failure. Because the probes and offsets depend only on $V$ and $Q$, the frames fetched and the choice to answer or abstain never depend on $c$, so changing $c$ moves only its own side, and a false claim is rejected whenever the order recovered from the frames is correct.

\section{Experiment}
\label{sec:exp}

\subsection{Experimental Setup}

\paragraph{Benchmark.}
We use the UCF-101 \citep{soomro2012ucf101} and NExT-QA \citep{xiao2021nextqa} datasets to build $500$ test videos per budget, each placing two brief UCF-101 action clips of distinct categories in two fixed slots over a NExT-QA background, with the placement order as ground truth so no model ever sets the truth. We draw each video's timing at random within ranges that give Interventions A and B their \textsc{on-grid} and \textsc{in-gap} offsets: slot spacing $\Delta\in[5,11]$~s, video length $L=k\Delta$, and event length $d_e\in[1.8,4.5]$~s with $d_e<0.7\,\Delta$, so the coverage ratio $\rho=d_e/\Delta<1$ and both offsets exist at every budget. We vary the frame budget over $k\in\{8,16,24,32,48,64\}$, with $k=16$ as the primary setting, raising $k$ at fixed $\Delta$ so the video lengthens and the two event frames dilute. A uniform grid meets both events or neither and never one alone, so the excluded single-event case and the full construction are provided in the appendix.

\paragraph{Models.}
We evaluate nine open models that differ in how each encodes frame time. LLaVA-OneVision-7B \citep{li2024llavaonevision}, the two LLaVA-NeXT-Video-7B variants (base and DPO) \citep{zhang2024llavanextvideo}, and InternVL3-8B \cite{zhu2025internvl3} and InternVL3.5-8B \citep{wang2025internvl35} take only a frame list. LLaVA-Video-7B \citep{zhang2024llavavideo} adds each frame's absolute time as a text instruction, Molmo2-8B \citep{clark2026molmo2} passes per-frame times through video metadata, and Qwen2.5-VL-7B \citep{bai2025qwen25vl} and Qwen3-VL-8B \citep{bai2025qwen3vl} bind time through a time-aware position encoding. To make sure order can be read only from the sampled frames, the main results withhold any native time channel, and the escape test adds it back. All experiments are conducted on two NVIDIA RTX 6000 Ada GPUs ($48$\,GB each).

\paragraph{Metrics.}
We report $S$, $C$, and $J=C-S$ as hard rates on the model's top option, together with the no-claim informedness $J_0=2q-1$, which is Youden's $J$ of the model with no claim on this balanced two-option task.

\begin{table*}[t]
\centering\small
\resizebox{\textwidth}{!}{%
\begin{tabular}{l cccc cccc cccc cccc}
\toprule
\multirow{2}{*}{Model} & \multicolumn{4}{c}{$k{=}16$, \textsc{on-grid}} & \multicolumn{4}{c}{$k{=}16$, \textsc{in-gap}} & \multicolumn{4}{c}{$k{=}32$, \textsc{on-grid}} & \multicolumn{4}{c}{$k{=}32$, \textsc{in-gap}} \\
\cmidrule(lr){2-5}\cmidrule(lr){6-9}\cmidrule(lr){10-13}\cmidrule(lr){14-17}
 & $J_0$ & $S$ & $C$ & $J$ & $J_0$ & $S$ & $C$ & $J$ & $J_0$ & $S$ & $C$ & $J$ & $J_0$ & $S$ & $C$ & $J$ \\
\midrule
LLaVA-OneVision-7B & $0.51$ & $0.93$ & $0.98$ & $0.05$ & $0.00$ & $0.95$ & $0.95$ & $0.00$ & $0.03$ & $0.98$ & $0.98$ & $0.01$ & $0.00$ & $0.98$ & $0.98$ & $0.00$ \\
LLaVA-Video-7B & $0.56$ & $0.95$ & $0.99$ & $0.05$ & $0.00$ & $0.98$ & $0.98$ & $0.00$ & $0.52$ & $0.95$ & $0.99$ & $0.04$ & $0.00$ & $0.97$ & $0.97$ & $0.00$ \\
InternVL3-8B & $0.52$ & $0.26$ & $0.77$ & $\mathbf{0.51}$ & $0.00$ & $0.59$ & $0.59$ & $0.00$ & $0.45$ & $0.28$ & $0.68$ & $\mathbf{0.41}$ & $0.00$ & $0.55$ & $0.55$ & $0.00$ \\
InternVL3.5-8B & $0.46$ & $0.82$ & $0.95$ & $0.12$ & $0.00$ & $0.92$ & $0.92$ & $0.00$ & $0.37$ & $0.85$ & $0.94$ & $0.09$ & $0.00$ & $0.91$ & $0.91$ & $0.00$ \\
Molmo2-8B & $0.65$ & $0.87$ & $0.99$ & $0.12$ & $0.00$ & $0.96$ & $0.96$ & $0.00$ & $0.60$ & $0.90$ & $0.97$ & $0.08$ & $0.00$ & $0.94$ & $0.94$ & $0.00$ \\
\midrule
LLaVA-NeXT-Video-7B & $0.00$ & $0.84$ & $0.84$ & $0.00$ & $0.00$ & $0.83$ & $0.83$ & $0.00$ & $0.00$ & $0.83$ & $0.83$ & $0.00$ & $0.00$ & $0.83$ & $0.83$ & $0.00$ \\
LLaVA-NeXT-Video-7B-DPO & $0.00$ & $0.88$ & $0.88$ & $0.00$ & $0.00$ & $0.87$ & $0.87$ & $0.00$ & $0.00$ & $0.89$ & $0.89$ & $0.00$ & $0.00$ & $0.88$ & $0.88$ & $0.00$ \\
Qwen2.5-VL-7B & $0.04$ & $0.93$ & $0.95$ & $0.02$ & $0.00$ & $0.92$ & $0.92$ & $0.00$ & $0.02$ & $0.90$ & $0.92$ & $0.02$ & $0.00$ & $0.91$ & $0.91$ & $0.00$ \\
Qwen3-VL-8B & $0.02$ & $0.97$ & $0.97$ & $0.00$ & $0.00$ & $0.97$ & $0.97$ & $0.00$ & $0.00$ & $0.98$ & $0.98$ & $0.00$ & $0.00$ & $0.98$ & $0.98$ & $0.00$ \\
\bottomrule
\end{tabular}}
\caption{Main results at $k=16$ and $k=32$ under a strong claim, with native time withheld.
}
\label{tab:main}
\end{table*}

\subsection{Results}
\label{sec:results}

\paragraph{The Sampling Gate.}
We test whether the sampled frames gate deference by scoring each model \textsc{on-grid} and \textsc{in-gap} at a fixed frame count, with and without the claim. Table~\ref{tab:main} collects $S$, $C$, and $J$ at $k=16$ and $k=32$. In the \textsc{in-gap} case the two twins present identical frames, so every model accepts a true and a false claim at the same rate, giving $S=C$ and $J=0$ for all nine, and a two one-sided test confirms the equivalence. In the \textsc{on-grid} case a model that reads the order can break this symmetry, and the no-claim informedness $J_0$ splits the nine into two classes: five read the order (\textsc{on-grid} $J_0$ from $0.46$ to $0.65$), while four cannot (\textsc{on-grid} $J_0$ from 0.00 to 0.04), so no effect can appear for them. Among the five order-readers, four still cave under a strong claim, accepting a false ordering $82$ to $95\%$ of the time for $J$ between $0.05$ and $0.12$; grounding alone does not buy resistance, since Molmo2 reads the order best ($J_0=0.65$) yet is among the most sycophantic. InternVL3 alone treats the claim as evidence, accepting a false ordering only $S=0.26$ and a true one $C=0.77$ for a gap $J=0.51$ that towers over every other model and proves the weighting ceiling is not fundamental. Availability is therefore necessary but not sufficient: sampling sets what the model can use, and its prior then decides how much survives the claim.

\paragraph{Frame Attribution.}
We test which frames carry the order and whether a model weights them, by replacing the two event frames with background, comparing a placebo that removes an equal number of background frames, and reading a background-only prior $b$ with the events gone. Table~\ref{tab:attr} gives the frame swap alongside the prior $b$. Removing the event frames drops informedness to exactly $0.00$ for all five order-readers (McNemar's test \citep{mcnemar1947note}, $p<0.001$), while the placebo leaves it near the \textsc{on-grid} value, so the collapse is caused by which frames are present at a fixed count. The prior $b$ tells the weighting side: with no order visible every order-reader already accepts the claim ($b\ge0.92$), except InternVL3 at $b=0.62$, whose events then suppress the false claim strongly ($S-b=-0.36$). InternVL3 alone weighs the evidence, while the others hit a ceiling set by how much they trust the user.

\begin{table}[t]
\centering\small
\resizebox{\columnwidth}{!}{%
\begin{tabular}{l cccccc}
\toprule
Model & $J_0$ & no events & placebo & $b$ & $C-b$ & $S-b$ \\
\midrule
LLaVA-OneVision-7B & $0.51$ & $0.00$ & $0.52$ & $0.94$ & $+0.04$ & $+0.00$ \\
LLaVA-Video-7B     & $0.56$ & $0.00$ & $0.62$ & $0.98$ & $+0.01$ & $-0.01$ \\
InternVL3-8B       & $0.52$ & $0.00$ & $0.54$ & $0.62$ & $+0.16$ & $\mathbf{-0.36}$ \\
InternVL3.5-8B     & $0.46$ & $0.00$ & $0.49$ & $0.92$ & $+0.04$ & $-0.06$ \\
Molmo2-8B          & $0.65$ & $0.00$ & $0.65$ & $0.96$ & $+0.04$ & $-0.07$ \\
\midrule
LLaVA-NeXT-Video-7B     & $0.00$ & $0.00$ & $0.00$ & $0.83$ & $+0.01$ & $+0.01$ \\
LLaVA-NeXT-Video-7B-DPO & $0.00$ & $0.00$ & $0.00$ & $0.87$ & $+0.02$ & $+0.01$ \\
Qwen2.5-VL-7B      & $0.04$ & $0.00$ & $0.07$ & $0.93$ & $+0.02$ & $-0.02$ \\
Qwen3-VL-8B        & $0.02$ & $0.00$ & $0.07$ & $0.97$ & $+0.00$ & $+0.01$ \\
\bottomrule
\end{tabular}%
}
\caption{Frame-attribution and weighting results.}
\label{tab:attr}
\end{table}

\paragraph{The Weighting Ceiling.}
We test how the gap moves with the claim and the scoring by crossing a strong and a mild claim with the hard and the confidence rate. Table~\ref{tab:sensitivity} holds the four resulting cells. Reading the confidence mass rather than the top choice widens the gap for every caving model (LLaVA-OneVision $0.05\to0.08$, LLaVA-Video $0.05\to0.07$), and a milder claim widens it again (LLaVA-OneVision to $0.10$), following \citet{atwell2025basil} in separating a blind shift toward the claim from a sound update. InternVL3 runs the other way, its gap shrinking from $0.51$ under a strong claim to $0.29$ under a mild one, because a firmer claim it can reject buys a larger correction. The four that cannot read the order stay flat. That the caving persists with both events in the frames confirms weighting as a separate cause, and it is why a ``trust the user less'' fix only trades caving for rejected true claims, never recovering the ordering evidence that the sampler failed to place among the frames it kept.

\begin{table}[t]
\centering\small
\resizebox{\columnwidth}{!}{%
\begin{tabular}{l cc cc}
\toprule
\multirow{2}{*}{Model} & \multicolumn{2}{c}{hard rate $J$} & \multicolumn{2}{c}{confidence rate $J$} \\
\cmidrule(lr){2-3}\cmidrule(lr){4-5}
& strong & mild & strong & mild \\
\midrule
LLaVA-OneVision-7B & $0.05$ & $0.07$ & $0.08$ & $0.10$ \\
LLaVA-Video-7B     & $0.05$ & $0.08$ & $0.07$ & $0.09$ \\
InternVL3-8B       & $0.51$ & $0.29$ & $0.43$ & $0.25$ \\
InternVL3.5-8B     & $0.12$ & $0.15$ & $0.10$ & $0.12$ \\
Molmo2-8B          & $0.12$ & $0.11$ & $0.13$ & $0.12$ \\
\midrule
LLaVA-NeXT-Video-7B     & $0.00$ & $0.00$ & $0.00$ & $0.00$ \\
LLaVA-NeXT-Video-7B-DPO & $0.00$ & $0.00$ & $0.00$ & $0.00$ \\
Qwen2.5-VL-7B      & $0.02$ & $0.03$ & $0.03$ & $0.04$ \\
Qwen3-VL-8B        & $0.00$ & $0.00$ & $0.00$ & $0.00$ \\
\bottomrule
\end{tabular}%
}
\caption{Claim strength and scoring results.}
\label{tab:sensitivity}
\end{table}

\paragraph{Absolute-Time Escape.}
We test whether absolute time can undo the sampling gate by adding the timestamp of every sampled frame back to each model through its own native channel. Table~\ref{tab:escape} reports the results, where the underlined entries denote the model's default setting. The outcome is the same in every case: the \textsc{in-gap} $J_0$ stays at exactly $0.00$ with and without time and is omitted from the table, because an \textsc{in-gap} sample holds no event frame and the times it carries belong to the background. This now holds for models that read order well, not only for weak ones: LLaVA-OneVision-7B and LLaVA-Video-7B still read the order well \textsc{on-grid}, with $J_0$ near $0.5$, yet gain nothing \textsc{in-gap} from the timestamps. In the \textsc{on-grid} case, the size of the effect varies with the model, but the \textsc{in-gap} value never does. Adding time moves the LLaVA models only a little (the strong-claim gap rises from $0.05$ to $0.09$ for LLaVA-OneVision and from $0.05$ to $0.10$ for LLaVA-Video) and barely changes InternVL3-8B, which alone treats the claim as evidence (\textsc{on-grid} $J_0$ from $0.52$ to $0.57$), yet it markedly helps InternVL3.5-8B, whose \textsc{on-grid} $J_0$ rises from $0.46$ to $0.73$. Molmo2-8B behaves the same way: adding its per-frame metadata times raises it from $0.65$ to $0.73$, so it does draw on those times. Qwen3-VL-8B and the two LLaVA-NeXT-Video-7B variants, which cannot read the order, sit at chance with or without time. Through all of this, the \textsc{in-gap} value stays exactly zero, because the times a model is given belong to the frames it sampled, which \textsc{in-gap} are all background. Absolute time therefore does not create availability: however it is supplied, and however much it helps \textsc{on-grid}, it can only refine an answer in the regime where the events were actually sampled, and never manufactures one from a sample holding neither of the two events.

\begin{table}[t]
\centering\small
\resizebox{\columnwidth}{!}{%
\begin{tabular}{l cc cc}
\toprule
\multirow{2}{*}{Model} & \multicolumn{2}{c}{\textsc{on-grid} $J_0$} & \multicolumn{2}{c}{\textsc{on-grid} $J$} \\
\cmidrule(lr){2-3}\cmidrule(lr){4-5}
 & no time & with time & no time & with time \\
\midrule
LLaVA-OneVision-7B      & $\underline{0.51}$ & $0.59$ & $\underline{0.05}$ & $0.09$ \\
LLaVA-Video-7B          & $0.56$ & $\underline{0.50}$ & $0.05$ & $\underline{0.10}$ \\
InternVL3-8B            & $\underline{0.52}$ & $0.57$ & $\underline{0.51}$ & $0.53$ \\
InternVL3.5-8B          & $\underline{0.46}$ & $0.73$ & $\underline{0.12}$ & $0.17$ \\
Molmo2-8B               & $0.65$ & $\underline{0.73}$ & $0.12$ & $\underline{0.12}$ \\
\midrule
LLaVA-NeXT-Video-7B     & $\underline{0.00}$ & $0.00$ & $\underline{0.00}$ & $0.00$ \\
LLaVA-NeXT-Video-7B-DPO & $\underline{0.00}$ & $0.00$ & $\underline{0.00}$ & $0.01$ \\
Qwen2.5-VL-7B           & $0.04$ & $\underline{0.05}$ & $0.02$ & $\underline{0.03}$ \\
Qwen3-VL-8B             & $0.02$ & $\underline{0.01}$ & $0.00$ & $\underline{0.00}$ \\
\bottomrule
\end{tabular}%
}
\caption{Absolute-time escape results.
}
\label{tab:escape}
\end{table}

\paragraph{Budget Dilution.}
We test how availability dilutes as the budget grows by raising $k$ at a fixed spacing, so an \textsc{in-gap} offset exists at every budget and the \textsc{in-gap} $J_0=0.00$ throughout while the two captured event frames shrink to a smaller share of the sample. Figure~\ref{fig:budget} traces the no-claim \textsc{on-grid} $J_0$ under different budgets. Dilution is real but strongly model-dependent. LLaVA-OneVision-7B loses availability between $k=16$ and $k=24$ (\textsc{on-grid} $J_0$ falling $0.51\to0.05$), whereas LLaVA-Video-7B and Molmo2-8B keep a high \textsc{on-grid} $J_0$ across every budget ($0.66\to0.49$ and $0.79\to0.57$) and InternVL3/3.5 hold it through $k=32$ (out of memory for a higher $k$). Deference under a strong claim dilutes along with availability: InternVL3's $J$ falls from $0.51$ at $k=16$ to $0.41$ at $k=32$ (Table~\ref{tab:main}). Some order-reading models keep the two events legible even as they shrink to a small fraction of the sample, while others lose them, and the models that cannot read the order (Qwen3-VL and both LLaVA-NeXT-Video variants) stay at chance at every $k$. Availability thus depends not only on whether the events are captured but on whether they still occupy a large enough fraction of the budget to be seen at all.

\begin{figure}[t]
\centering
\includegraphics[width=\columnwidth]{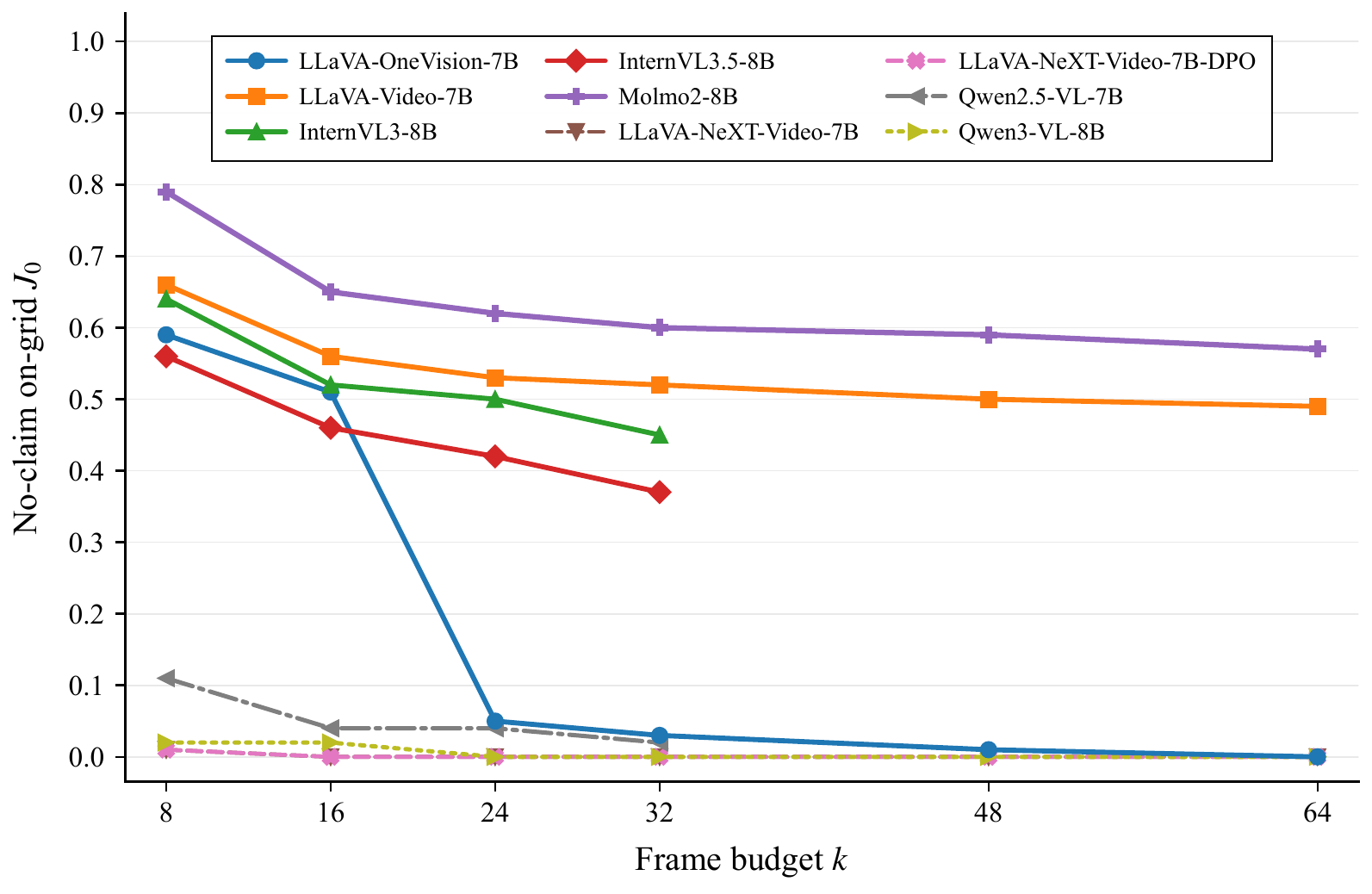}
\caption{Budget dilution results as the frame budget $k$ grows.
}
\label{fig:budget}
\end{figure}

\paragraph{Reversal Test.}
We run the reversal test of Algorithm~\ref{alg:policy}, which answers only when the forward-reverse gap $|R|$ exceeds $\theta$ and otherwise resamples denser or abstains. Table~\ref{tab:policy} compares it to two always-answer baselines. \emph{1 offset} reads the order from a single sampled grid, and \emph{$k{=}32$} from a uniform grid at the largest budget. For the test, its \emph{coverage} is the fraction of videos it answers, its \emph{mean $k$} the average budget it spends, and its \emph{accuracy} the order accuracy on answered videos. A single offset misses the events about half the time, so its accuracy sits near chance, whereas the test reaches $0.92$--$1.00$ accuracy on the five order-reading models at a mean budget below $k=32$. On the four that cannot read the order it abstains on $87$ to $94\%$ of videos rather than invent an order, since their forward and reversed readings agree and $|R|$ rarely passes the threshold. We choose $\theta=0.3$ by a grid search (in the appendix), the value that best balances coverage against accuracy for the order-reading models.

\begin{table}[t]
\centering\small
\resizebox{\columnwidth}{!}{%
\begin{tabular}{lccccc}
\toprule
\multirow{2}{*}{Model} & \multicolumn{2}{c}{always answer} & \multicolumn{3}{c}{reversal test ($\theta=0.3$)} \\
\cmidrule(lr){2-3}\cmidrule(lr){4-6}
& 1 offset & $k{=}32$ & mean $k$ & coverage & accuracy \\
\midrule
LLaVA-OneVision-7B & $0.52$ & $0.51$ & $15.2$ & $0.47$ & $0.92$ \\
LLaVA-Video-7B     & $0.53$ & $0.75$ & $21.4$ & $0.74$ & $0.97$ \\
InternVL3-8B       & $0.53$ & $0.71$ & $20.6$ & $0.70$ & $0.97$ \\
InternVL3.5-8B     & $0.52$ & $0.65$ & $21.2$ & $0.63$ & $1.00$ \\
Molmo2-8B          & $0.55$ & $0.78$ & $16.9$ & $0.85$ & $0.95$ \\
\midrule
LLaVA-NeXT-Video-7B     & $0.50$ & $0.49$ & $28.4$ & $0.06$ & $0.33$ \\
LLaVA-NeXT-Video-7B-DPO & $0.50$ & $0.50$ & $20.8$ & $0.06$ & $0.20$ \\
Qwen2.5-VL-7B      & $0.50$ & $0.52$ & $21.1$ & $0.13$ & $0.82$ \\
Qwen3-VL-8B        & $0.51$ & $0.50$ & $15.1$ & $0.09$ & $0.43$ \\
\bottomrule
\end{tabular}%
}
\caption{Reversal test results.}
\label{tab:policy}
\end{table}

\section{Discussion}

\paragraph{Availability and Weighting.}
A false claim can go uncorrected for two reasons, that the ordering evidence was never sampled (availability) or that it was sampled and still outweighed by the user (weighting). The two demand different fixes, and only weighting is what a ``trust the user less'' remedy can touch. Our test gives a simple rule: a model that reads the order at a high budget yet still caves has failed on weighting rather than availability, and no calibration of trust recovers evidence that was never sampled. Treating this sampling failure as a weighting problem is what misdirects the usual fixes toward trusting the user even less.

\paragraph{Against a Better Sampler.}
A missing-frame problem invites a smarter selector, but prompt-guided selectors such as AKS \citep{tang2025aks} and FRAG \citep{huang2025frag} read the user's claim and can be steered to fetch frames that support it, which turns the selector into part of the attack \citep{cao2025poisonvid}. Our test instead reverses the sampled frames to recover only the order they support without reading the claim, and it sits on top of any selector. With a prompt-guided selector, it becomes a diagnosis rather than a baseline: if caving still matches blind agreement, the selector has not removed the problem, and if it hides the problem, it does so only by trusting the very claim the test is built to check.

\paragraph{Scope of the Effect.}
The coverage ratio $\rho$ makes the claim testable. A brief event can fall between two sampled frames only when coverage is sparse, that is when $\rho<1$, and the effect grows as coverage drops. When coverage is dense, a uniform grid always samples both events, so availability holds and any deference that remains is a well-studied weighting problem, where our cause predicts no change. The effect therefore depends on sparse coverage rather than on long videos. A long video with many frames still captures both events, and a short clip misses them only when too few frames fall on the pair.

\paragraph{Limitations.}
The usable videos are narrow, since a probe needs an order question, two brief swappable events, a budget sparse enough that an offset can miss both, and a model that reads the order when they are sampled. The effect is bounded on two sides, by how strongly a model already trusts the user and by its order-reading competence. Our tested models are heavily sycophantic, so under an explicit claim the \textsc{on-grid} gap is small even when the order is available, though it is clearer in the confidence rate and under a milder claim. The best order-readers read the order well, while weaker models sit at chance even with both events sampled. The reversal test inherits the first bound. It detects an unsampled event, while an \textsc{on-grid} model that has the evidence and still caves goes undetected. Our events are trimmed action clips rather than steps of a real procedure, and we study order alone, though the reversal recipe extends to any judgment with a transform that flips its ground truth while leaving the prior unchanged.

\section{Conclusion}
In this paper, we show that for order questions, whether a video language model caves to a false claim or rejects a true one turns on whether its sampled frames contain both events. Two interventions at a fixed budget establish this, a reorder that flips the claim's truth and an offset shift that captures or misses the events, with $\rho<1$ bounding the effect. Missing the events makes the twins identical, so every model reaches $J=0$, while capturing them lets an order-reading model recover part of the gap, bounded by how much it trusts the user. This separates availability from weighting, two causes the literature has treated as one, and shows that no calibration of trust recovers unsampled evidence. The safe response is to detect that the frames do not fix the order by reversing them, then resample or abstain without reading the claim.

\bibliography{refs}

\clearpage
\appendix
\section{Appendix}
\label{app:tech}

\subsection{Experimental Setup Details}

\paragraph{Construction.}
Events are trimmed single-action clips from UCF-101 and backgrounds are NExT-QA videos. We form each test video by pairing two clips of distinct action categories over a looped NExT-QA background, and we drop visually similar category pairs so the two events are clearly distinguishable and their order carries no prior. The only requirement is that the two categories be distinct, since the order of two unrelated actions cannot be guessed without watching the video. This yields the $500$ test videos we build per budget, with the placement order as ground truth, so no model and no annotator is required to set the truth. For each test video, we build twin $+$ by placing the first clip in the earlier slot and the second in the later slot, and twin $-$ by swapping them, then sample $k$ frames as input to the models.

\paragraph{Order-Reading Competence.}
We keep every test video rather than pre-selecting the ones a model can read by first passing them at a large frame budget. We instead report, per model, the \textsc{on-grid} informedness $J_0=2q-1$ that the model reaches when it answers the order question with no user claim, which directly shows whether it can read the order when both events are sampled. The \textsc{in-gap} case needs no such check. Because no sampled frame lands in either event slot, the two twins present identical frames, so the \textsc{in-gap} order sits at chance and any \textsc{in-gap} difference between the twins is zero by construction.

\paragraph{Implementation.}
We draw each video's timing per instance: the frame spacing is $\Delta\in[5,11]$ seconds, the background length is $L=k\Delta$, each event lasts $d_e\in[1.8,4.5]$ seconds with $d_e<0.7\,\Delta$, and the two event slots start at $p_1$ and $p_2=p_1+m\Delta$ ($m$ is a small positive integer), so that both fall on the sampling grid. This keeps the coverage ratio $\rho=d_e/\Delta<1$ fixed while the budget varies, so an \textsc{on-grid} and an \textsc{in-gap} offset exist at every $k\in\{8,16,24,32,48,64\}$ (primary $k=16$). Raising $k$ dilutes rather than removes the \textsc{in-gap} regime. The answer is read from the next-token logits for the option letters, averaged over the two letter assignments to balance position, with deterministic decoding. Each model's native time channel (LLaVA-Video's text timestamps, Molmo's video metadata, the Qwen models' position encoding) is withheld in the main table and restored only in the escape test, where the models that have no native channel are instead given a text-instruction timestamp like the one LLaVA-Video uses by default.

\paragraph{Controls and Statistics.}
Beyond the no-claim baseline $J_0$, we report a background-only run $b$ with the two events removed, which is identical on the two twins and cancels in $J$, through the gains $C-b$ and $S-b$. Because background frames separate the two events, there is no visible cut between the clips for a model to exploit. Significance uses McNemar's test \citep{mcnemar1947note} run separately for the false-claim and true-claim directions, and a two one-sided test for the \textsc{in-gap} equivalence to zero at margin $|J|<0.1$. We did not run a separate human study. On inspection the \textsc{on-grid} frames show two distinct actions in temporal order, and the strongest models read that order at up to $0.83$ accuracy ($J_0=0.65$).

\subsection{Additional Analyses}

\paragraph{Why Most Models Cave.}
The \textsc{on-grid} deference gap $J$ in the main experiment is small for four of the five order-reading models and large only for InternVL3. Two quantities explain the split. The first is availability, which the no-claim \textsc{on-grid} informedness $J_0$ measures. Five models read the order and four stay at chance (Table~\ref{tab:prior}). The second is a model's prior deference to the user, read from the \textsc{in-gap} accept rate under a strong claim, where no order is visible and the model can only follow the claim. Every model accepts the false claim \textsc{in-gap} at a high rate, from $0.83$ to $0.98$, with one exception: InternVL3 at $0.59$. The \textsc{on-grid} gap tracks this prior: a model whose prior already accepts the claim has no room left for the evidence to move it, so $J$ stays near zero even when $J_0$ is high, while InternVL3, whose prior is far from saturated, turns its availability into a large $J$. Caving is therefore a saturated prior rather than a grounding failure, since Molmo reads the order best of all nine yet is among the most sycophantic. The background-only rate $b$ in Table~\ref{tab:attr} confirms it, at $b\ge0.92$ for the caving readers there and $0.62$ for InternVL3. The fix that helps is a weaker prior rather than a better sampler, and the reversal test reaches the same end by never reading the claim at all.

\begin{table}[t]
\centering\small
\resizebox{\columnwidth}{!}{%
\begin{tabular}{l ccc}
\toprule
Model & \textsc{on-grid} $J_0$ & \textsc{in-gap} $S$ & \textsc{on-grid} $J$ \\
\midrule
LLaVA-OneVision-7B & $0.51$ & $0.95$ & $0.05$ \\
LLaVA-Video-7B & $0.56$ & $0.98$ & $0.05$ \\
InternVL3-8B & $\mathbf{0.52}$ & $\mathbf{0.59}$ & $\mathbf{0.51}$ \\
InternVL3.5-8B & $0.46$ & $0.92$ & $0.12$ \\
Molmo2-8B & $0.65$ & $0.96$ & $0.12$ \\
\midrule
LLaVA-NeXT-Video-7B & $0.00$ & $0.83$ & $0.00$ \\
LLaVA-NeXT-Video-7B-DPO & $0.00$ & $0.87$ & $0.00$ \\
Qwen2.5-VL-7B & $0.04$ & $0.92$ & $0.02$ \\
Qwen3-VL-8B & $0.02$ & $0.97$ & $0.00$ \\
\bottomrule
\end{tabular}%
}
\caption{Availability, prior, and deference results.}
\label{tab:prior}
\end{table}

\paragraph{Single Sampled Event.}
The integer spacing $p_2=p_1+m\Delta$ makes a uniform grid meet both event slots or neither, so the main construction never samples exactly one event. We still probe this excluded case directly. At a fixed budget $k$, we keep one event's frames and replace the other's with background, in an early variant that keeps the earlier event and a late variant that keeps the later one, and read the no-claim order informedness $J_0$ as before. Table~\ref{tab:onesamp} reports all nine models, with each value averaged over $500$ probes and $\bar{J_0}$ averaging the early and late variants.

The four models that cannot read the order stay at $J_0\approx0$ in both single-event variants, so removing one event adds no order cue of its own and these models act as a negative control. Among the five order-reading models, InternVL3 and InternVL3.5 also fall to $J_0\approx0$, since a single event tells them nothing about order. LLaVA-OneVision and Molmo2 instead give an equal and opposite response, positive when the early event is kept and negative when the late event is kept ($+0.50$ and $-0.44$ for LLaVA-OneVision), which shows that they read the order from where the one visible event falls in the sample rather than from a comparison of two events. This position cue carries no order information overall, and once we average over which event is kept, every model returns $\bar{J_0}\le0.05$. A single event thus gives at most a guess from position that averages out over offsets, so the both-events conclusion of the main tables holds. Still, one deployment caveat remains. A model that reads the order from a single event's position can answer confidently against the true order when only the later event is in frame, and the reversal test is built to catch this instability.

\begin{table}[t]
\centering\small
\setlength{\tabcolsep}{4pt}
\resizebox{\columnwidth}{!}{%
\begin{tabular}{l cccc}
\toprule
Model & \textsc{on-grid} $J_0$ & early $J_0$ & late $J_0$ & $\bar{J_0}$ \\
\midrule
LLaVA-OneVision-7B & $0.50$ & $+0.50$ & $-0.44$ & $+0.03$ \\
LLaVA-Video-7B & $0.55$ & $+0.08$ & $-0.06$ & $+0.01$ \\
InternVL3-8B & $0.52$ & $+0.04$ & $0.00$ & $+0.02$ \\
InternVL3.5-8B & $0.47$ & $0.00$ & $+0.04$ & $+0.02$ \\
Molmo2-8B & $0.66$ & $+0.26$ & $-0.17$ & $+0.05$ \\
\midrule
LLaVA-NeXT-Video-7B & $0.00$ & $+0.01$ & $-0.01$ & $0.00$ \\
LLaVA-NeXT-Video-7B-DPO & $0.00$ & $+0.03$ & $-0.03$ & $0.00$ \\
Qwen2.5-VL-7B & $0.04$ & $+0.03$ & $-0.03$ & $0.00$ \\
Qwen3-VL-8B & $0.02$ & $+0.03$ & $-0.03$ & $0.00$ \\
\bottomrule
\end{tabular}%
}
\caption{Single-sampled-event results. $\bar{J_0}$ averages the early and late variants.}
\label{tab:onesamp}
\end{table}

\paragraph{Threshold Sensitivity.}
We choose the threshold $\theta$ of Table~\ref{tab:policy} by a grid search on a separate set of videos, to meet a target coverage. Table~\ref{tab:theta} varies $\theta$ and reports the five order-reading models as one combined group. Raising $\theta$ trades coverage for accuracy monotonically, from $0.90$ coverage at $0.88$ accuracy to $0.44$ coverage at $0.99$ accuracy, so accuracy stays high across the range and $\theta=0.3$ strikes a good balance, at $0.67$ coverage and $0.96$ accuracy. On the four models that cannot read the order, coverage collapses toward zero at every $\theta$, since the test finds no order to report and abstains. The accuracy on the few clips it still answers rests on too small a sample to be meaningful.

\begin{table}[t]
\centering\small
\setlength{\tabcolsep}{4pt}
\begin{tabular}{l ccc c cc}
\toprule
& \multicolumn{3}{c}{can read order} & & \multicolumn{2}{c}{cannot read order} \\
\cmidrule(lr){2-4}\cmidrule(lr){6-7}
$\theta$ & coverage & accuracy & mean $k$ & & coverage & accuracy \\
\midrule
$0.1$ & $0.90$ & $0.88$ & $14.5$ & & $0.43$ & $0.56$ \\
$0.2$ & $0.77$ & $0.94$ & $16.5$ & & $0.23$ & $0.60$ \\
$0.3$ & $0.67$ & $0.96$ & $19.3$ & & $0.09$ & $0.57$ \\
$0.4$ & $0.54$ & $0.98$ & $21.5$ & & $0.05$ & $0.67$ \\
$0.5$ & $0.44$ & $0.99$ & $23.4$ & & $0.03$ & $0.78$ \\
\bottomrule
\end{tabular}
\caption{Reversal test across thresholds $\theta$, combined within each model group.}
\label{tab:theta}
\end{table}

\subsection{Additional Discussion}

\paragraph{Practical Stakes.}
The two directions matter most where mistakes are costly. In incident review, surveillance, and clinical video, a model that accepts a confident wrong claim and a model that rejects a correct human correction are both dangerous, and a fix that only teaches the model to trust the user less trades one error for the other \citep{beigi2025smart,li2025mmsy,rahman2025pendulum}. When the real problem is missing evidence rather than misplaced trust, the right answer is to detect that the sampled frames do not determine the order and either look again or decline to answer, rather than to argue with the user. A test that never reads the claim keeps the valid corrections that one-sided fixes throw away.

\end{document}